\documentclass[english,runningheads,a4paper]{llncs}
\usepackage[T1]{fontenc}
\usepackage[latin9]{inputenc}
\usepackage[b5j]{geometry}
\usepackage{float}
\usepackage{graphicx}

\makeatletter

\floatstyle{ruled}
\newfloat{algorithm}{tbp}{loa}
\providecommand{\algorithmname}{Algorithm}
\floatname{algorithm}{\protect\algorithmname}

\newtheorem{thm}{\protect\theoremname}
\newtheorem{defn}[thm]{\protect\definitionname}
\newtheorem{lem}[thm]{\protect\lemmaname}
\newtheorem{prop}[thm]{\protect\propositionname}

\makeatother

\usepackage{babel}
\providecommand{\definitionname}{Definition}
\providecommand{\lemmaname}{Lemma}
\providecommand{\propositionname}{Proposition}
\providecommand{\theoremname}{Theorem}

\begin{document}

\title{A Parameterized Approximation Algorithm for The Shallow-Light Steiner
Tree Problem
\thanks{This project was supported by the Natural Science
Foundation of Fujian Province (2012J05115), Doctoral Fund of Ministry
of Education of China for Young Scholars (20123514120013) and Fuzhou
University Development Fund (2012-XQ-26). }
}
\author{{Longkun Guo$^{1}$, Kewen Liao$^{2}$}, XiuJun Wang$^{1}$}
\institute{{$^{1}$School of Mathematics and Computer Science,
Fuzhou University, China}\\
{$^{2}$School of Computer Science, University of
Adelaide, Australia}}

\maketitle
\begin{abstract}
For a given graph $G=(V,\, E)$ with a terminal set $S$ and a selected
root $r\in S$, a positive integer cost and a delay on every edge
and a delay constraint $D\in Z^{+}$, the shallow-light Steiner tree
(\emph{SLST}) problem is to compute a minimum cost tree spanning the
terminals of $S$, in which the delay between root and every vertex
is restrained by $D$. This problem is NP-hard and very hard to approximate.
According to known inapproximability results, this problem admits
no approximation with ratio better than factor $(1,\, O(\log^{2}n))$
unless $NP\subseteq DTIME(n^{\log\log n})$ \cite{khandekar2013some},
while it admits no approximation ratio better than $(1,\, O(\log|V|))$
for $D=4$ unless $NP\subseteq DTIME(n^{\log\log n})$ \cite{bar2001generalized}.
Hence, the paper focus on parameterized algorithm for \emph{SLST}.
We firstly present an exact algorithm for \emph{SLST} with time complexity
$O(3^{|S|}|V|D+2^{|S|}|V|^{2}D^{2}+|V|^{3}D^{3})$, where $|S|$ and
$|V|$ are the number of terminals and vertices respectively. This
is a pseudo polynomial time parameterized algorithm with respect to
the parameterization: ``number of terminals''. Later, we improve
this algorithm such that it runs in polynomial time $O(\frac{|V|^{2}}{\epsilon}3^{|S|}+\frac{|V|^{4}}{\epsilon}2^{|S|}+\frac{|V|^{6}}{\epsilon})$
, and computes a Steiner tree with delay bounded by $(1+\epsilon)D$
and cost bounded by the cost of an optimum solution, where $\epsilon>0$
is any small real number. To the best of our knowledge, this is the
first parameterized approximation algorithm for the \emph{SLST} problem. \end{abstract}

\section{Introduction}

The well-known shallow-light Steiner tree problem (or namely the delay
restrained minimum Steiner tree problem) is defined as below:
\begin{defn}
For a graph $G=(V,\, E)$ with a terminal set $S$, a root vertex
$r\in S$, a cost function $c:\, E\rightarrow Z^{+}$, a delay function
$d:\, E\rightarrow Z^{+}$, and a delay bound $D\in Z^{+}$, the\emph{
shallow-light Steiner tree (SLST) problem} is to compute a minimum
cost Steiner tree $slst$ spanning all terminals of $S$, such that
the delay from $r$ to every terminal in $slst$ is not larger than
$D$.
\end{defn}
For notation briefness, we assume $|V|=n,\,|E|=m,\,|S|=t$ in graph
$G$, and use \emph{SLST }and\emph{ $slst$} to denote the shallow-light
Steiner tree problem and an optimal shallow-light Steiner tree respectively.
 For the \emph{SLST} problem, bifactor approximation algorithms have
been developed.
\begin{defn}
An algorithm $A$ is a bifactor $\left(\alpha,\,\beta\right)$-approximation
for the \emph{SLST} problem, if and only if for every instance of
\emph{SLST, $A$} computes a Steiner tree $slst$ in polynomial time,
such that the delay from $r$ to every terminal in $slst$ is bounded
by $\alpha*D$ and the cost of $Slst$ is bounded by $\beta$ times
of the cost of the optimal solution.

Noting that single factor $\beta$-approximation is identical to bifactor
$\left(1,\,\beta\right)$-approximation for \emph{SLST}, we use them
interchangeably in the text.
\end{defn}
\textbf{Related Work.} It is known that the \emph{SLST} problem is
NP-hard, and can not be approximated better than factor $(1,\, O(\log^{2}n))$
unless $NP\subseteq DTIME(n^{\log\log n})$ \cite{khandekar2013some}.
This is because the group Steiner tree problem can be embedded into
this problem. Furthermore, no polylogarithmic approximation within
polynomial time complexity has been developed. The best work is a
long standing result due to Charikar et al, which is a polylogarithmic
approximation in quasi-polynomial time, i.e. factor-$O(\log^{2}t)$
approximation within time complexity $n^{O(\log t)}$\cite{charikar1998approximation}.
Due to the difficulty in single factor approximation algorithm design,
bifactor approximation has been investigated. Hajiaghayi et al presented
an $(O(\log^{2}t),\, O(\log^{4}t))$-approximation algorithm that
runs in polynomial time \cite{hajiaghayi2006approximating}. Besides,
Kapoor and Sarwat gave an approximation with bifactor $(O(\frac{p\log t}{\log p}),O(\frac{\log t}{\log p}))$,
where $p$ is an input parameter \cite{kapoor2007bounded}. The last
algorithm is an approximation that improves the cost of the tree,
and is with bifactor $(O(t),\, O(1))$ when $p=t$ \cite{kapoor2007bounded}.

The \emph{SLST }problem remains hard to approximate even when $S=V$.
In that case, this problem becomes the shallow light spanning tree
(SLT) problem, which has broad applications in network design, VLSI
and etc. For computational complexity, the SLT problem is claimed
to be with inapproximability hardness of $(1,\,\Omega(\log n))$ \cite{naor1997improved}.
For approximation, Charikar et al's $O(\log^{2}n)$ ratio with time
complexity $n^{O(\log n)}$\cite{charikar1998approximation} is still
the best single factor result. Naor and Schieber gave an approximation
bifactor of $(2,\, O(\log n))$, i.e. with delay and cost bounded
by 2 times and $O(\log n)$ times of that of the optimal solution
respectively \cite{naor1997improved}. To the best of our knowledge,
these are the best long standing approximation ratios. Some special
cases of the SLT problem are also interesting. If edge cost is equal
to the delay for each edge, the SLT problem remains NP-hard and admit
no approximation algorithms with bifactor $(\alpha,\,\beta)$ for
any $\alpha>1$ and $1\leq\beta<1+\frac{2}{\alpha-1}$ \cite{khuller1995balancing},
while the best possible result for \emph{SLST} is a $(1+\epsilon,\, O(\log(\frac{1}{\epsilon})))$-approximation
\cite{elkin2011steiner}.  Moreover, the SLT problem remains NP-hard
when all edge delays are equal, but polynomially solvable when all
edge costs are equal \cite{salama1997delay}. For the equal-delay
case, namely the hop constrained minimum spanning tree problem, Althaus
et al have presented an approximation with a ratio of $(1,\, O(\log n))$
in \cite{althaus2005approximating}.

Another two important special cases of the \emph{SLST} problem is
when $D$ is constant or when all edge delays are equal. Unfortunately,
for the former case, \emph{SLST} can not be approximated better than
a factor of $(1,\, O(\log n))$ for even $D=4$ unless $NP\subseteq DTIME(n^{\log\log n})$
\cite{bar2001generalized}, since the Set Cover problem can be embedded
into this case. Bar-Ila et al also developed a factor-$(1,\, O(\log n))$
approximation for the cases of $D=4,\,5$ in the same paper, achieving
the best possible ratio. When all edge delays are equal, namely the
hop constrained minimum Steiner tree problem, it is open that if there
exists factor-$(1,\, O(\log n))$ approximation for this problem,
as the spanning case.

\textbf{Our Contribution. }The first result of this paper is an exact
algorithm, with time complexity $O(3^{|S|}|V|D+2^{|S|}|V|^{2}D^{2}+|V|^{3}D^{3})$,
for the \emph{SLST} problem. This result indicates that if the number
of terminal and the delay constraint are bounded, the \emph{SLST}
problem is polynomial solvable. Our technique is mainly based on constructing
an auxiliary graph, where every Steiner tree satisfies the delay constraint,
i.e. in the auxiliary graph, we only need to compute Steiner tree
without considering the delay constraint. Though its time complexity
seems terrible, the exact algorithm is efficient for real-world applications
for $|S|<80$, particularly when $D=o(n)$ or all edge delays are
equal (the hop constrained minimum Steiner tree problem).

On the theoretical side, we note that this algorithm runs in pseudo
polynomial time (for constant $|S|$), since $D$ appears in the formula
of the time complexity. The second result is to improve this time
complexity to polynomial time $O(\frac{|V|^{2}}{\epsilon}3^{|S|}+\frac{|V|^{4}}{\epsilon}2^{|S|}+\frac{|V|^{6}}{\epsilon})$
following a similar line of polynomial-time approximation scheme (PTAS)
design, such that it computes a Steiner tree with delay bounded by
$(1+\epsilon)D$ and cost bounded by the cost of an optimum solution.

\section{A Parameterized Approximation Algorithm for the Shallow Light Steiner
Problem}

In this section we shall approximate the shallow-light Steiner tree
(\emph{SLST}) problem. Firstly and intuitively, our main observation
is that the difficulty of computing a \emph{$slst$} comes from obeying
the given delay constraint. Therefore, our key idea is to construct
an auxiliary directed graph $H$ where there exists only cost (i.e.
no delay) on edges, such that every Steiner tree (spanning the same
terminal set) in $H$ corresponds to a Steiner tree that satisfies
the given delay constraint $D$ in $G$. Secondly since the directed
Steiner tree problem is known parameterized tractable with respect
to the parameterization: ``number of terminals''\cite{guo2011parameterized,ding2007finding},
an exact algorithm is immediately obtained; then an approximation
algorithm with ratio $(1+\epsilon,\,1)$ can be derived from the exact
algorithm by a method of shrinking the value of $D$. The approximation
algorithm computes a $slst$ with delay bounded by $D(1+\epsilon)$
and cost bounded by the cost of an optimum $slst$.

\subsection{Construction of the Auxiliary Graph}

Though different in technique details, the key idea to construct the
auxiliary graph is similar to the auxiliary graph used to balance
the cost and delay of $k$ disjoint shortest paths in \cite{guo2013improved}:
using layer graphs. For a given graph $G=(V,E)$ with positive integer
cost and delay on every edge, and a delay constraint $D$, the layer
graph $H$, i.e. the auxiliary graph to be constructed, contains vertices,
terminals and edges roughly as in the following:
\begin{enumerate}
\item $D$ vertices $v_{l}^{1},\dots,v_{l}^{D}$ corresponding to every
vertex $v_{l}\in G$ ;
\item $D-d(e)$ edges $\left\langle v_{j}^{1},\, v_{l}^{d(e)+1}\right\rangle ,\,\dots,\left\langle v_{j}^{D-d(e)},\, v_{l}^{D}\right\rangle $
corresponding to every edge $e=\left\langle v_{j},v_{l}\right\rangle \in E$
and with $c(\left\langle v_{j}^{i},\, v_{l}^{d(e)+i}\right\rangle )=c(e)$;
\item one terminal $v_{l}$, corresponding to every terminal $v_{l}\in S\subseteq G$,
together with cost-0 edges $\{\left\langle v_{l}^{i},\, v_{l}\right\rangle \vert i=1,\dots,D\}$
that connect auxiliary vertices of $v_{l}$ to the auxiliary terminal;
\end{enumerate}
Therefore, $H$ has $O(|V|*D)$ vertices, $O(|E|*D)$ edges, and $|S|$
terminals. The construction is formerly as in Algorithm \ref{alg:Construction-of-auxiliary}
(An example of such construction is as depicted in Figure \ref{fig:Construction-of-acyclic}).

\begin{algorithm}
\textbf{Input}: Graph $G=(V,E)$, a set of terminals $S\subseteq V$,
a root vertex $r\in S$, cost $c:\, e\rightarrow Z^{+}$and delay
$d:\, e\rightarrow Z^{+}$ on every edge $e\in E$, and a delay constraint
$D$;

\textbf{Output}: Auxiliary acyclic graph $H$ and the terminal set
therein, $S_{H}$.
\begin{enumerate}
\item $H:=\{r\}$, $S_{H}:=\{r\}$;
\item \textbf{For }each $v_{l}\in V\setminus\{r\}$ \textbf{do}

\begin{enumerate}
\item $H:=H\cup\{v_{l}^{1},\dots,v_{l}^{D}\}$;
\item \textbf{If} $v_{l}\in S$ \textbf{then}

\begin{enumerate}
\item $H:=H\cup\{v_{l}\}\cup\{\left\langle v_{l}^{i},\, v_{l}\right\rangle \vert i=1,\dots,D\}$,
and set $c(\left\langle v_{l}^{i},\, v_{l}\right\rangle ):=0$ for
each $i$;
\item $S_{H}:=S_{H}\cup\{v_{l}\}$;
\end{enumerate}
\end{enumerate}
\item \textbf{For} each $e=\left\langle v_{j},\, v_{l}\right\rangle \in E$
that $r\notin e$ \textbf{do}

\quad{}$H:=H\cup\{\left\langle v_{j}^{i},\, v_{l}^{d(e)+i}\right\rangle \vert i=1,\,\dots,\, D-d(e)\}$,
and set $c(\left\langle v_{j}^{i},\, v_{l}^{d(e)+i}\right\rangle ):=c(e)$
for each $i$;

\item \textbf{For} each $e=\left\langle r,v_{l}\right\rangle $ \textbf{do}

\quad{}$H:=H\cup\{\left\langle r,v_{l}^{d(e)}\right\rangle \}$,
and set $c(\left\langle r,v_{l}^{d(e)}\right\rangle ):=c(e)$ .

\item Return $H$ and $S_{H}$.
\end{enumerate}
\caption{\label{alg:Construction-of-auxiliary}Construction of auxiliary graph
$H$.}
\end{algorithm}

\begin{figure}
\begin{centering}
\includegraphics[width=0.7\textwidth]{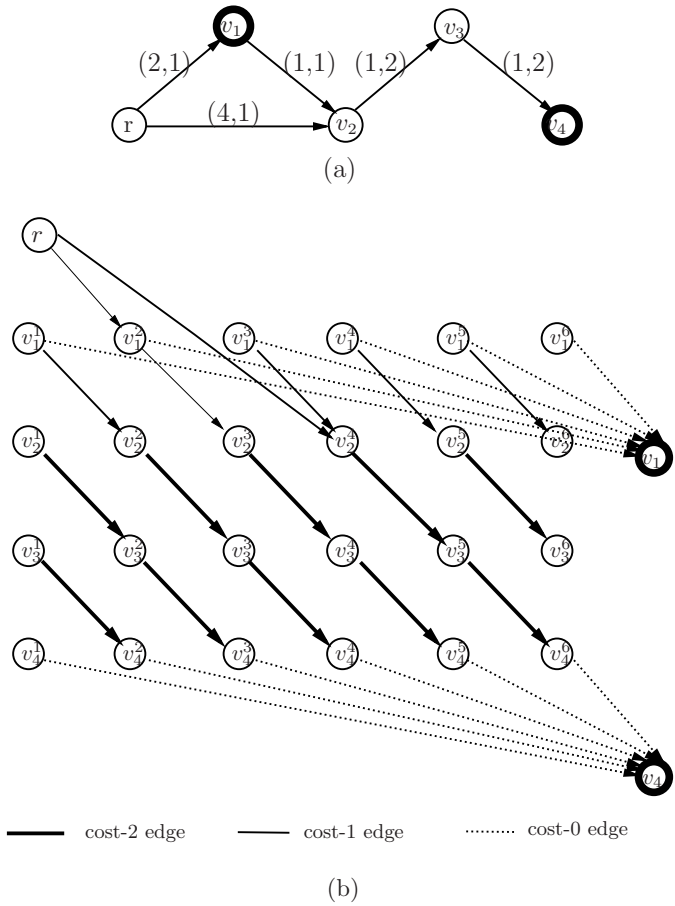}
\par\end{centering}

\caption{\label{fig:Construction-of-acyclic}Construction of acyclic graphs:
(a) is the original graph, in which $r,\, v_{1},\, v_{4}$ are terminals;
(b) is the constructed auxiliary graph, in which $r,\, v_{1},\, v_{4}$
are terminals.}
\end{figure}

It remains to show that the $r$-rooted minimum cost directed Steiner
tree in $H$ corresponds to a $r$-rooted minimum $slst$ in $G$.
\begin{lem}
\label{lem:degree}A minimum cost directed Steiner tree rooted at
$r$ in $H$ contains at most one vertex of $\{v_{l}^{1},\dots,v_{l}^{D}\}$
for each $l$.\end{lem}
\begin{proof}
Let $R$ be a $r$-rooted minimum cost directed Steiner tree in $H$.
Suppose $R$ contains $v_{l}^{j}$ and $v_{l}^{j+\Delta}$. Then we
show that $R$ is not minimum and get a contradiction. Let $R'$ be
$R$ except removing the edge entering $v_{l}^{j+\Delta}$ and replacing
every edge in the subtree of $R$ that roots at $v_{l}^{j+\Delta}$,
say $\left\langle v_{h}^{i+\Delta},v_{h'}^{i'+\Delta}\right\rangle $
by edge $\left\langle v_{h}^{i},v_{h'}^{i'}\right\rangle $. Apparently,
$R'$ spanning the same terminal set as $R$. That is, there exists
a directed Steiner tree $R'$ with less cost than $R$ in $H$. This
contradicts with the fact that $R$ is minimum.\end{proof}
\begin{thm}
\label{thm:acyclic}Let $S_{H}$ be the set of terminal vertices $\{v_{1},\dots,v_{t}\}$
in $H$. Then there exists a $r$-rooted directed Steiner tree spanning
$S_{H}$ of minimum cost $C$ in $H$ iff there exists a Steiner tree
spanning $S$ of minimum cost $C$ with delay between $r$ and every
terminal restrained by $D$ in $G$.\end{thm}
\begin{proof}
Let $R$ be a minimum cost directed Steiner tree rooted at $r$ in
$H$. Let $R'$ be a subgraph of $G$, in which $e(v_{j},v_{l})\in R'$
if and only if there exists $e(v_{j}^{i_{j}},v_{l}^{i_{l}})\in R$.
Then because $c(v_{j},v_{l})=c(v_{j}^{i_{j}},v_{l}^{i_{l}})$, we
have $c(R')=c(R)$. It remains to show $R'$ is a Steiner tree. From
Lemma \ref{lem:degree}, $|\{v_{l}^{1},\dots,v_{l}^{D}\}\cap R|\leq1$
holds for every $l$. So a path connecting $r$ to a terminal in $H$
corresponds to a path connecting $r$ to a terminal in $G$.  Then
since every terminal of $S_{H}$ is reachable from $r$ in $R$, all
terminals of $S$ are connected to $r$ in $R'$. Besides, because
$R$ is a tree, $R'$ contains no loops or parallel edges. Therefore,
$R'$ is a Steiner tree of $G$.

Let $R'$ be a Steiner tree in $G$. Then there is a unique path from
root $r$ to every other vertex of $R'$. Hence, every vertex of $R'$
has a unique delay from $r$. Let $R$ contains edge $(v_{j}^{d(v_{j})},v_{j})$
for every $v_{j}\in S_{H}$, and edge $(v_{j}^{d(v_{j})},v_{l}^{d(v_{j})+d(v_{j},v_{l})})$
if and only if $(v_{j},v_{l})\in R'$, where $d(v_{j})$ is the delay
from $r$ to $v_{j}$ in $R$ and $d(v_{j},v_{l})$ the delay from
$v_{j}$ to $v_{l}$. Since the delay of from $r$ to every vertex
in $R'$ is not larger than $D$, edge $(v_{j}^{d(v_{j})},v_{l}^{d(v_{j})+d(v_{j},v_{l})})$
belongs to $H$, and hence $R\subseteq H$. Then because every $v_{j}\in R$
is reachable from $r$ and no loop or parallel edge exists following
the construction of $R$, $R$ is a Steiner tree in $H$ with cost
$c(R)\leq c(R')$. This completes the proof.
\end{proof}

\subsection{A parameterized Approximation Algorithm for Shallow-Light Steiner
Tree }

This subsection shall give an exact algorithm and a parameterized
approximation algorithm for the \emph{SLST} problem. From Theorem
\ref{thm:acyclic}, an algorithm for the \emph{SLST} problem can be
obtained by computing a minimum cost directed Steiner tree in $H$.
Unfortunately, it is known that the (minimum) directed Steiner tree
problem is NP-hard and maybe even more difficult to approximate than
\emph{SLST}, i.e. only a quasi-polynomial time algorithm with a polylogarithmic
approximation factor has been developed\cite{charikar1998approximation}.
However, when the number of the terminals is a constant, the directed
Steiner tree problem is polynomial solvable, as stated in the proposition
below:
\begin{prop}
\label{prop:An-optimum-solution}\cite{guo2011parameterized}An optimum
solution to the directed Steiner tree problem can be computed within
$O(3^{t}n+2^{t}n^{2}+n^{3})$, where $t$ and $n$ are the number
of terminals and vertices respectively.
\end{prop}
Following Algorithm \ref{alg:Construction-of-auxiliary}, Theorem
\ref{thm:acyclic} and Proposition \ref{prop:An-optimum-solution},
we could now state the exact algorithm for the \emph{SLST} problem
as in the following:

\begin{algorithm}
\textbf{Input}: Graph $G=(V,E)$, $S\subseteq V$, $r\in S$, cost
function $c:\, e\rightarrow Z^{+}$and delay function $d:\, e\rightarrow Z^{+}$
, a delay constraint $D$, and auxiliary graph $H$ with $S_{H}$,

\textbf{Output}: $R'$, an optimum solution to the \emph{SLST} problem.
\begin{enumerate}
\item $R':=\emptyset$;
\item Compute a minimum cost Steiner tree in $H$, say $R$ spanning the
terminal of $S_{H}$ by the method of \cite{ding2007finding};
\item \textbf{For} every $e(v_{j}^{i_{j}},v_{l}^{i_{l}})\in R$ \textbf{do}

\quad{}\textbf{If} $e(v_{j},v_{l})\notin R'$ \textbf{then} $R':=R'\cup\{e(v_{j},v_{l})\}$;

\item Return $R'$.
\end{enumerate}
\caption{\label{alg:An-exact-algorithm}An exact algorithm for \emph{SLST}.}

\end{algorithm}

Following Theorem \ref{thm:acyclic} and Proposition \ref{prop:An-optimum-solution},
we immediately have the correctness of Algorithm \ref{alg:An-exact-algorithm}.
For time complexity, since $H$ contains $O(m*D)$ edges, $O(n*D)$
vertices and $t$ terminals, it takes $O(3^{t}nD+2^{t}n^{2}D^{2}+n^{3}D^{3})$
time to compute a minimum Steiner tree in $H$. Hence, we have:
\begin{thm}
\label{thm:best-1}Algorithm \ref{alg:An-exact-algorithm} solved
the SLST problem correctly, and runs in time $O(3^{t}nD+2^{t}n^{2}D^{2}+n^{3}D^{3})$.
\end{thm}
We note that Algorithm \ref{alg:An-exact-algorithm} runs in pseudo-polynomial
time, since the formula of the time complexity contains $D$. However,
following the technique of polynomial-time approximation scheme (PTAS)
design, a parameterized approximation algorithm for the \emph{SLST}
problem could proceed as: firstly compute $G'$, which is $G$ except
the delay of every edge $e$ is sat to $\left\lfloor \frac{n*d(e)}{\epsilon*D}\right\rfloor $,
such that the value of delay constraint is shrunken from $D$ to a
polynomial on $n$; secondly construct graph $H$ with the new delay
on edges; and finally run Algorithm \ref{alg:An-exact-algorithm}
on the auxiliary graph $H$ of the new delay. Formally, the parameterized
approximation algorithm for the\emph{ SLST }problem is as in the following:

\begin{algorithm}
\textbf{Input}: A given parameter $\epsilon$, graph $G=(V,E)$, $S\subseteq V$,
$r\in S$, cost $c:\, e\rightarrow Z^{+}$and delay $d:\, e\rightarrow Z^{+}$
on every edge $e\in E$, and a delay constraint $D$;

\textbf{Output}: $R"$, an approximation solution to the \emph{SLST}
problem.
\begin{enumerate}
\item \textbf{For} every edge of $G$ \textbf{do}

\quad{}$d(e):=\left\lfloor \frac{n*d(e)}{\epsilon*D}\right\rfloor $;

/{*} Compute $G'$.{*}/

\item Construct auxiliary graph $H$ and compute $S_{H}$ using Algorithm
\ref{alg:Construction-of-auxiliary};
\item Compute a minimum cost Steiner tree $R"$ subjected to the new delay
constraint $\left\lfloor \frac{n}{\epsilon}\right\rfloor $by applying
Algorithm \ref{alg:An-exact-algorithm} on $G$ and $H$ with respect
to the new delay;
\item Return $R"$.
\end{enumerate}
\caption{\label{alg:appro-algorithm-1}A parameterized approximation algorithm
for \emph{SLST}.}
\end{algorithm}

Following Algorithm \ref{alg:appro-algorithm-1}, the delay constraint
in $G'$ is $\left\lfloor \frac{n}{\epsilon}\right\rfloor $. Then
from Lemma \ref{thm:best-1}, the time complexity of the algorithm
is $O(3^{t}\frac{n^{2}}{\epsilon}+2^{t}\frac{n^{4}}{\epsilon}+\frac{n^{6}}{\epsilon})$
after shrinking $D$ to $O(\frac{n}{\epsilon})$. Hence, we have
\begin{lem}
\label{lm:best}Algorithm \ref{alg:appro-algorithm-1} runs in time
$O(3^{t}\frac{n^{2}}{\epsilon}+2^{t}\frac{n^{4}}{\epsilon}+\frac{n^{6}}{\epsilon})$.
.
\end{lem}
It remains to show the approximation of the algorithm, which is given
by the following theorem:
\begin{thm}
\label{thm:ratio}Algorithm \ref{alg:appro-algorithm-1} computes
a Steiner tree spanning all terminals of $S$ in $G$ with cost bounded
by the cost of an optimum $slst$, and delay bounded by $(1+\epsilon)D$.\end{thm}
\begin{proof}
Clearly, an optimum $slst$ in $G$ will satisfy the new delay constraint
$\left\lfloor \frac{n}{\epsilon}\right\rfloor $ in $G'$. Then since
$R"$ is a optimum solution to \emph{SLST} in $G'$, it is with cost
not larger than the cost of an optimum $slst$ in $G$.

It remains to show the delay of $R"$ in $G$. Let $P$ be an arbitrary
path in $R"$. Then since the delay of $R"$ in $G'$ is bounded by
$\left\lfloor \frac{n}{\epsilon}\right\rfloor $, we have:

\begin{equation}
\sum_{e\in P}\left\lfloor \frac{n*d(e)}{\epsilon*D}\right\rfloor \leq\left\lfloor \frac{n}{\epsilon}\right\rfloor \label{eq:1}
\end{equation}

Following the definition of $\left\lfloor \,\right\rfloor $, $\frac{n*d(e)}{\epsilon*D}<1+\left\lfloor \frac{n*d(e)}{\epsilon*D}\right\rfloor $
holds, and hence:

\begin{equation}
\sum_{e\in P}\frac{n*d(e)}{\epsilon*D}<\sum_{e\in P}(\left\lfloor \frac{n*d(e)}{\epsilon*D}\right\rfloor +1)\label{eq:2}
\end{equation}

Combining Inequality (\ref{eq:1}) and (\ref{eq:2}) yields:

\begin{equation}
\sum_{e\in P}\frac{n*d(e)}{\epsilon*D}<\sum_{e\in P}\left\lfloor \frac{n*d(e)}{\epsilon*D}\right\rfloor +\sum_{e\in P}1\leq\left\lfloor \frac{n}{\epsilon}\right\rfloor +n\label{eq:3}
\end{equation}

Therefore, following Inequality (\ref{eq:3}), the delay of $R"$
in $G$ is:
\[
\sum_{e\in P}d(e)=\frac{\epsilon D}{n}\sum_{e\in P}\frac{n*d(e)}{\epsilon*D}<\frac{\epsilon D}{n}*(\left\lfloor \frac{n}{\epsilon}\right\rfloor +n)=(1+\epsilon)D.
\]

This completes the proof.
\end{proof}

\section{Conclusion}

This paper investigated exact algorithms and then parameterized approximation
algorithms for the SLST problem. The first result is an exact algorithm
that computes optimum $slst$ in time $O(3^{t}nD+2^{t}n^{2}D^{2}+n^{3}D^{3})$,
and the second result is a factor-$(1+\epsilon,\,1)$ approximation
algorithm with time complexity $O(3^{t}\frac{n^{2}}{\epsilon}+2^{t}\frac{n^{4}}{\epsilon}+\frac{n^{6}}{\epsilon})$.
A problem remained open is whether design of algorithms for the SLST
problem with polylogarithmic approximation ratio is possible.

\bibliographystyle{plain}
\bibliography{drmst}

\begin{thebibliography}{10}

\bibitem{althaus2005approximating}
Ernst Althaus, Stefan Funke, Sariel Har-Peled, Jochen Konemann, Edgar~A Ramos,
  and Martin Skutella.
\newblock Approximating< i> k</i>-hop minimum-spanning trees.
\newblock {\em Operations Research Letters}, 33(2):115--120, 2005.

\bibitem{bar2001generalized}
Judit Bar-Ilan, Guy Kortsarz, and David Peleg.
\newblock Generalized submodular cover problems and applications.
\newblock {\em Theoretical Computer Science}, 250(1):179--200, 2001.

\bibitem{charikar1998approximation}
M.~Charikar, C.~Chekuri, T.~Cheung, Z.~Dai, A.~Goel, S.~Guha, and M.~Li.
\newblock Approximation algorithms for directed steiner problems.
\newblock In {\em Proceedings of the ninth annual ACM-SIAM symposium on
  Discrete algorithms}, pages 192--200. Society for Industrial and Applied
  Mathematics, 1998.

\bibitem{ding2007finding}
Bolin Ding, J~Xu~Yu, Shan Wang, Lu~Qin, Xiao Zhang, and Xuemin Lin.
\newblock Finding top-k min-cost connected trees in databases.
\newblock In {\em Data Engineering, 2007. ICDE 2007. IEEE 23rd International
  Conference on}, pages 836--845. IEEE, 2007.

\bibitem{elkin2011steiner}
M.~Elkin and S.~Solomon.
\newblock Steiner shallow-light trees are exponentially lighter than spanning
  ones.
\newblock In {\em Foundations of Computer Science (FOCS), 2011 IEEE 52nd Annual
  Symposium on}, pages 373--382. IEEE, 2011.

\bibitem{guo2011parameterized}
Jiong Guo, Rolf Niedermeier, and Ondrej Such{\`y}.
\newblock Parameterized complexity of arc-weighted directed steiner problems.
\newblock {\em SIAM Journal on Discrete Mathematics}, 25(2):583--599, 2011.

\bibitem{guo2013improved}
Longkun Guo, Hong Shen, and Kewen Liao.
\newblock Improved approximation algorithms for computing k disjoint paths
  subject to two constraints.
\newblock In {\em Computing and Combinatorics}, pages 325--336. Springer, 2013.

\bibitem{hajiaghayi2006approximating}
M.T. Hajiaghayi, G.~Kortsarz, and M.~Salavatipour.
\newblock Approximating buy-at-bulk and shallow-light k-steiner trees.
\newblock {\em Approximation, Randomization, and Combinatorial Optimization.
  Algorithms and Techniques}, pages 152--163, 2006.

\bibitem{kapoor2007bounded}
S.~Kapoor and M.~Sarwat.
\newblock Bounded-diameter minimum-cost graph problems.
\newblock {\em Theory of Computing Systems}, 41(4):779--794, 2007.

\bibitem{khandekar2013some}
Rohit Khandekar, Guy Kortsarz, and Zeev Nutov.
\newblock On some network design problems with degree constraints.
\newblock {\em Journal of Computer and System Sciences}, 2013.

\bibitem{khuller1995balancing}
S.~Khuller, B.~Raghavachari, and N.~Young.
\newblock Balancing minimum spanning trees and shortest-path trees.
\newblock {\em Algorithmica}, 14(4):305--321, 1995.

\bibitem{naor1997improved}
J.~Naor and B.~Schieber.
\newblock Improved approximations for shallow-light spanning trees.
\newblock In {\em Foundations of Computer Science, 1997. Proceedings., 38th
  Annual Symposium on}, pages 536--541. IEEE, 1997.

\bibitem{salama1997delay}
H.F. Salama, D.S. Reeves, and Y.~Viniotis.
\newblock The delay-constrained minimum spanning tree problem.
\newblock In {\em Computers and Communications, 1997. Proceedings., Second IEEE
  Symposium on}, pages 699--703. IEEE, 1997.

\end{thebibliography}

\end{document}